\documentclass[12pt]{article}
\usepackage{amsmath}
\usepackage{graphicx,psfrag,epsf}
\usepackage{enumerate}
\usepackage{natbib}
\usepackage{hyperref}
\newcommand{\blind}{0}
\def\code#1{\texttt{#1}}
\newcommand{\proglang}[1]{\textsf{#1}}
\usepackage{accents}
\usepackage{amssymb}
\usepackage{bm}
\usepackage{graphicx}
\usepackage{natbib}
\usepackage{amsmath}
\makeatletter
\newcommand{\distas}[1]{\mathbin{\overset{#1}{\kern\z@\sim}}}%
\newsavebox{\mybox}\newsavebox{\mysim}
\newcommand{\distras}[1]{%
  \savebox{\mybox}{\hbox{\kern3pt$\scriptstyle#1$\kern3pt}}%
  \savebox{\mysim}{\hbox{$\sim$}}%
  \mathbin{\overset{#1}{\kern\z@\resizebox{\wd\mybox}{\ht\mysim}{$\sim$}}}%
}
\usepackage{breqn}
\usepackage{multirow}
\usepackage{bm}
\usepackage{multirow}
\usepackage[normalem]{ulem}
\useunder{\uline}{\ul}{}
\usepackage{graphicx}
\usepackage{wrapfig}
\usepackage{lscape}
\usepackage{rotating}
\usepackage{epstopdf}
\usepackage{comment}

\usepackage[symbol]{footmisc}

\usepackage[utf8]{inputenc}
\usepackage[english]{babel}
 \usepackage{amsthm}

 \newtheorem{theorem}{Theorem}

\usepackage{graphicx}
\usepackage{subcaption}

\usepackage{booktabs,subcaption,amsfonts,dcolumn}
\begin{document}

\def\spacingset#1{\renewcommand{\baselinestretch}%
{#1}\small\normalsize} \spacingset{1}


\if0\blind
{
  \title{\bf A Spatial Bayesian Semiparametric Mixture Model for Positive Definite Matrices with Applications to Diffusion Tensor Imaging}

  \author{
  Zhou Lan\\
    Department of Statistics, North Carolina State University\\
    and \\
    Brian J Reich \\
    Department of Statistics, North Carolina State University\\
    and\\
     Dipankar Bandyopadhyay \\
      Department of Biostatistics, Virginia Commonwealth University
    }

  \maketitle
} \fi

\if1\blind
{
  \bigskip
  \bigskip
  \bigskip
  \begin{center}
    {\LARGE\bf Title}
\end{center}
  \medskip
} \fi

\bigskip
\begin{abstract}
Diffusion tensor imaging (DTI) is a popular magnetic resonance imaging technique used to characterize microstructural changes in the brain. DTI studies quantify the diffusion of water molecules in a voxel using an estimated $3\times3$ symmetric positive definite diffusion tensor matrix. Statistical analysis of DTI data is challenging because the data are positive definite matrices. Matrix-variate information is often summarized by a univariate quantity, such as the fractional anisotropy (FA), leading to a loss of information. Furthermore, DTI analyses often ignore the spatial association of neighboring voxels, which can lead to imprecise estimates. Although the spatial modeling literature is abundant, modeling spatially dependent positive definite matrices is challenging. To mitigate these issues, we propose a matrix-variate Bayesian semiparametric mixture model, where the positive definite matrices are distributed as a mixture of inverse Wishart distributions with the spatial dependence captured by a Markov model for the mixture component labels. Conjugacy and the double Metropolis-Hastings algorithm result in fast and elegant Bayesian computing. Our simulation study shows that the proposed method is more powerful than non-spatial methods. We also apply the proposed method to investigate the effect of cocaine use on brain structure. The contribution of our work is to provide a novel statistical inference tool for DTI analysis by extending spatial statistics to matrix-variate data.

\end{abstract}

\noindent%

\spacingset{1.5}
\newpage
\section{Introduction} 
\label{sec:intro}
Measurement of signal attenuation from water diffusion, often considered one of the most important magnetic resonance contrast mechanisms \citep{alexander2007diffusion}, is usually achieved via diffusion tensor imaging (DTI) that maps and characterizes the 3-D diffusion of water molecules as a function of the spatial location \citep{basser1994mr}. The diffusion process in the brain reflects interactions with many obstacles, such as fibers, thereby revealing microscopic details about the underlying tissue architecture. Unlike ordinary images where scalars are summarized for each voxel, a distinguishing feature of DTI is each voxel is associated with a $3\times 3$ symmetric positive definite matrices which can be interpreted as the covariance matrix of a 3-D Gaussian distribution modeling the local Brownian motion of the water molecules \citep{schwartzman2008inference}. These positive definite matrices are also called the diffusion tensors (DTs). One important clinical application of the DTI is to detect regions of local differences in the brain between two groups (i.e., normal versus disease), revealing anatomical structural differences \citep{lo2010diffusion}. For example, the motivating data for this paper comes from a clinical DTI study \citep{ma2017preliminary}, where the scientific objective is to detect regions of differences between cocaine users and non-cocaine users.

Statistical analysis of DTI data is challenging due to the difficulty of modeling matrix-variate responses. One option is to project the DTs into fractional anisotropy (FA), a scalar describing the degrees of anisotropy of a diffusion process. However, some information is lost because different positive definite matrices may produce the same FA \citep{ennis2006orthogonal}. Matrix-variate methods potentially avoid information loss. There are relatively few matrix-variate methods available to analyze DTI data, and they can be broadly classified into the (inverse) Wishart matrix methods \citep{dryden2009non} and the random ellipsoid models \citep{schwartzman2008inference,lee2017inference}. However, these voxel-level models ignore information from neighboring voxels that may have similar neuronal activity (see Figure \ref{fig:typical_dti}), despite recommendations of incorporating this non-negligible spatial association to achieve efficient and valid inference \citep{spence2007accounting} as well as studies revealing that the disease status at proximally-located/neighboring voxels can be similar \citep[see][]{wu2013mapping,xue2018bayesian}. This motivates us to develop an improved spatial statistical model which (a) utilizes full matrix information, (b) captures spatial dependence, and (c) can be implemented via fast and elegant computing. 

\begin{figure}[ht!]
    \centering
    \includegraphics[width=0.32\textwidth]{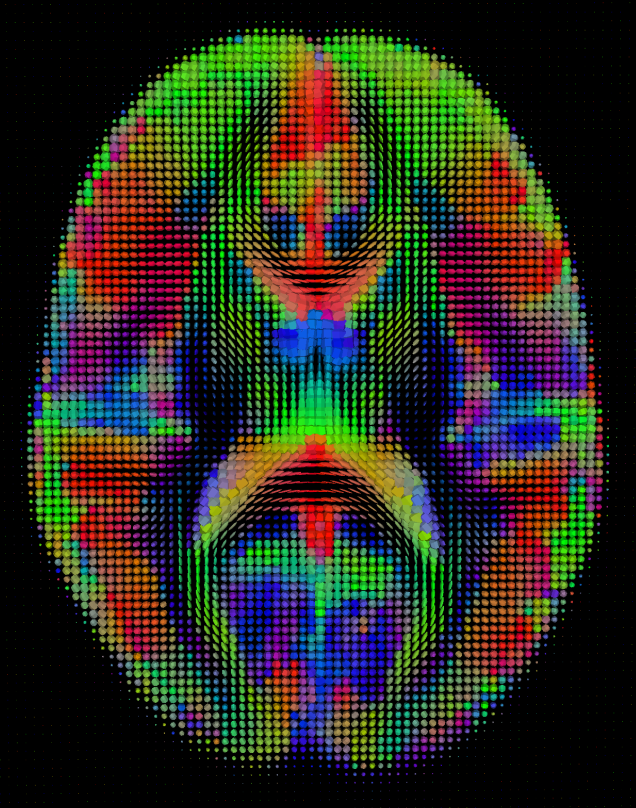}
    \caption{Diffusion tensor imaging of a human brain. Positive definite matrices which are visualized as ellipsoids are produced for each voxel, revealing the anatomical structure of the brain. The ellipsoids are generated using the visualization software “BrainSuite” (\url{http://brainsuite.org/}).}
    \label{fig:typical_dti}
\end{figure}

The spatial neuroimaging toolbox for univariate responses is considerably rich: \citet{woolrich2004fully} proposed a fully Bayesian model for spatiotemporal imaging data; \citet{kang2011meta} implemented spatial point processes for meta-analysis of imaging data; To select essential biological features, \citet{musgrove2016fast} introduced spatial Bayesian variable selection for neuroimaging data; Recently, \citet{reich2018fully} proposed spectral methods for ADNI data to provide computational benefits. All of these studies demonstrated an improvement in the precision of estimates by properly accounting for spatial dependence. 

In this vein, the Potts model, a generalization of the Ising model in statistical mechanics, has also been successfully applied to imaging \citep{johnson2013bayesian,li2018bayesian}. A desirable property of the Potts model is that it avoids smoothing over abrupt changes in the image intensity \citep{johnson2013bayesian}, and this makes it more attractive than available Gaussian kernel methods. To this end, we assume the positive definite DTs follow a mixture of inverse Wishart distributions, with the mixture component labels modeled via a (spatial) Potts model, representing a discrete Markov random field. This \textit{semiparametric mixture} specification refers to a class of flexible mixture distributions with a finite number of components \citep{lindsay1995review}. 

Besides spatial modeling, another important topic in neuroimaging is detecting regions of differences between two groups. Previous attempts at identifying regions of differences between two groups were formulated through voxel-wise hypothesis testing \citep{schwartzman2008inference,lee2017inference}. An alternative option is to construct multilevel hierarchical modeling accounting both subjective-level and group-level variation and use the group-level parameters for voxel-wise hypothesis testing \citep{woolrich2004multilevel, liu2014functional}. In this paper, we use the latter approach via extending the latent classic Potts model into a hierarchical \textit{two-way} framework, allowing hypothesis testing via group-level parameters and inter-subject variability simultaneously.


Our proposal is implemented using the Bayesian approach, accounting for the uncertainty of model parameters in all levels of the hierarchy. However, the Bayesian approach is often problematic in neuroimaging because of its heavy computational burden  \citep{cohen2017computational}. Although the associated Markov chain Monte Carlo (MCMC) algorithm is mostly composed of computationally tractable Gibbs steps that can be paralleled, a major drawback of the Potts model is the intractable normalizing constant, creating a bottleneck for hyperparameter updates. In this paper, it is resolved via the double Metropolis-Hastings algorithm \citep{liang2010double} recommended by \citet{park2018bayesian}.

To the best of our knowledge, this is the first work on exploring spatial associations in modeling positive definite matrix-variate data under a Bayesian semiparametric framework, with applications to DTI. In the rest of the paper, we first introduce the single-subject and multi-subject model, and the group hypothesis testing framework in Section \ref{sec:model}. Relevant MCMC computational details appear in Section \ref{sec:computing}. To demonstrate the improvement in performance compared to plausible alternatives, we perform simulation studies in Section \ref{sec:sim}. In Section \ref{sec:application}, we present the application to the motivating cocaine data set. Finally, Section \ref{sec:dis} concludes with a discussion.

\section{Model}
\label{sec:model}
In this section, we introduce the spatial Bayesian semiparametric mixture model for positive definite matrices. We introduce the single-subject model first in Section \ref{sec:single} and then extend to the multi-subject model in Section \ref{sec:multi}.

\subsection{Single-Subject Model}
\label{sec:single}
Let $\bm{A}_v$ be the $p\times p$ DT at voxel $v\in \{1, 2, .., n\}$. To ensure $\bm{A}_{v}$ is symmetric and positive definite, it is usually parameterized as a (inverse) Wishart matrix \citep{dryden2009non} or a Gaussian symmetric matrix-variate distribution \citep{schwartzman2008inference}. In this paper, we assume that $\bm{A}_{v}$ follows an inverse Wishart distribution as
\begin{equation}
    \bm{A}_v|\bm{M}_{v},m\distas{indep.} \mathcal{IW}_p(\bm{M}_v,m),
\end{equation}
where $\mathcal{IW}_p(\bm{M}_v,m)$ is the inverse Wishart distribution parameterized (Appendix A) to have mean $\bm{M}_v$ and degrees of freedom $m>p+1$, and the DTs are independently distributed across $v$ given the mean matrices $\bm{M}_v$ and the degrees of freedom $m$. The mean matrices are modeled as a finite mixture of Wishart distributions, denoted as $[\bm{M}_v|g_v=k]:= \bm{V}_k$ where $g_v\in \{1,2,...,K\}$ is the latent cluster label. The prior of $\bm{V}_k$ is $\bm{V}_k\distas{i.i.d} \mathcal{W}_p(\bm{\Sigma},\nu)$ where $\mathcal{W}_p(\bm{\Sigma},\nu)$ is the Wishart distribution parameterized (Appendix A) to have mean $\bm{\Sigma}$ and degrees of freedom $\nu>p$.

Spatial dependence of the DTs is achieved through the dependence of the mean matrices $\bm{M}_v$. We induce spatial dependence via the latent cluster labels that follow a weighted Potts model, specified via the full conditional distributions:
\begin{equation}
    \begin{aligned}
    \mathcal{P}_k=P(g_{v}=k|\beta, \eta_k, \bm{g}_{-v})\propto \exp\left[\eta_k+\beta\sum_{u\in N_v}\mathcal{I}(g_{u}=k)\right],
    \end{aligned}
\end{equation}
where $\bm{g}_{-v}$ is the full set $\bm{g}=\{g_1, g_2, ..., g_n\}$ excluding $g_v$, $N_v$ is a set of indices of the neighboring voxels of $v$, and $\mathcal{I} [\mathcal{A}] = 1$ if event $\mathcal{A}$ is true and $\mathcal{I} [\mathcal{A}] = 0$ otherwise. Given $\bm{g}_{-v}$ but marginal over $g_v$, the distribution of $\bm{A}_v$ is the mixture of $K$ inverse Wishart distributions:
$$\sum_{k=1}^K \mathcal{P}_k \mathcal{IW}_p(\bm{A}_v|\bm{V}_k,m),$$
where $\mathcal{IW}_p(\bm{A}|\bm{V},m)$ is the inverse Wishart density function of $\bm{A}$ with the mean matrix $\bm{V}$ and the degrees of freedom $m$. Therefore, this semiparametric mixture model spans a rich class of density functions.

Via the Potts model, an image can be considered as a network whose nodes are the voxels. In this network, every voxel is connected to its neighboring voxels. The full conditional distribution of $g_v$ depends only on the voxels in the neighboring set $N_v$ and therefore the process is Markovian. Since the spatial parameter $\beta$ is the coefficient of the neighboring term $\sum_{u\in N_v}\mathcal{I}(g_{u}=k)$, the spatial parameter $\beta$ controls the dependence on the neighboring voxels. 

Unlike the classic Potts model \citep{wu1982potts}, the terms $\eta_k$ are added as offset terms controlling the overall mass on each cluster. We set $\eta_k=-k^\xi$ so the parameter $\xi>0$ is the concentration parameter controlling the homogeneity of the latent cluster labels. It is problematic to pre-specify the number of components $K$ in a mixture model \citep{mccullagh2008many} but the offset terms provide more weight on the key components and fewer weights on the trivial components. We fit the model by setting $K$ to be an upper bound on the number of active clusters and allow the data to determine the number of active clusters via estimation of $\xi$: if $\xi\rightarrow 0$, there are several active clusters; if $\xi$ is large, there are a few active clusters. As a result, the model is less sensitive to the number of components $K$ when the offset term $\eta_k$ is included. This claim is verified in the simulation studies (Section \ref{sec:sim}) and the real data application (Section \ref{sec:application}), where we find similar results for different $K$.

Quantifying spatial dependence is a vital issue in spatial statistics and neuroimaging. Since this model is for matrix-variate data, we use the expected squared Frobenius norm to measure dependence. The dependence between matrices $\bm{A}$ and $\bm{B}$ can be summarized as $\mathbb{E}||\bm{A}-\bm{B}||_F^2=\mathbb{E}Tr[(\bm{A}-\bm{B})^T(\bm{A}-\bm{B})]$. The norm increases as dependence decreases. If $\bm{A}$ and $\bm{B}$ are $1\times 1$, the expected squared Frobenius norm is the classic variogram \citep{cressie1992statistics} of spatial statistics. In this regard, the expected squared Frobenius norm can be treated as the variogram for matrix-variate data and useful in measuring spatial dependence. In the rest of the paper, we simply call the expected squared Frobenius norm the variogram.

For the Potts model described above, the variogram is $$\mathcal{V}(u, v)=\mathbb{E}||\bm{A}_{u}-\bm{A}_{v}||_F^2= \gamma(m,\nu,\bm{\Sigma}) P(u,v|\beta,\xi),$$ where $P(u,v|\beta,\xi)$ is the marginal (over all other cluster labels $\bm{g}$) probability of $g_u\not=g_v$, and $\gamma(m,\nu,\bm{\Sigma})$ is a measure of the variability in $\bm{A}_v|\bm{M}_v$ and variability of $\bm{V}_k$ across $K$. Therefore, the multivariate spatial dependence structure is separable \citep{cressie1992statistics} in that the dependence is the product of a non-spatial term $\gamma(m,\nu,\bm{\Sigma})$ that controls cross dependence and a spatial term $P(u,v|\beta,\xi)$ that controls spatial dependence.

We give the expression of $\gamma(m,\nu,\bm{\Sigma})$ in Appendix B. When $p=3$ and $\bm{\Sigma}=\bm{I}$, the non-spatial term has an expression that is $ \frac{12(m+\nu-4)(2m-7)}{\nu(m-3)(m-6)}$ where $m>6$ and $\nu>3$. Therefore, in this special case, the cross dependence decreases if $m$ or $\nu$ is larger (See Figure \ref{fig:density}).

\begin{figure}[ht!]
    \centering
    \includegraphics[width=0.4\textwidth]{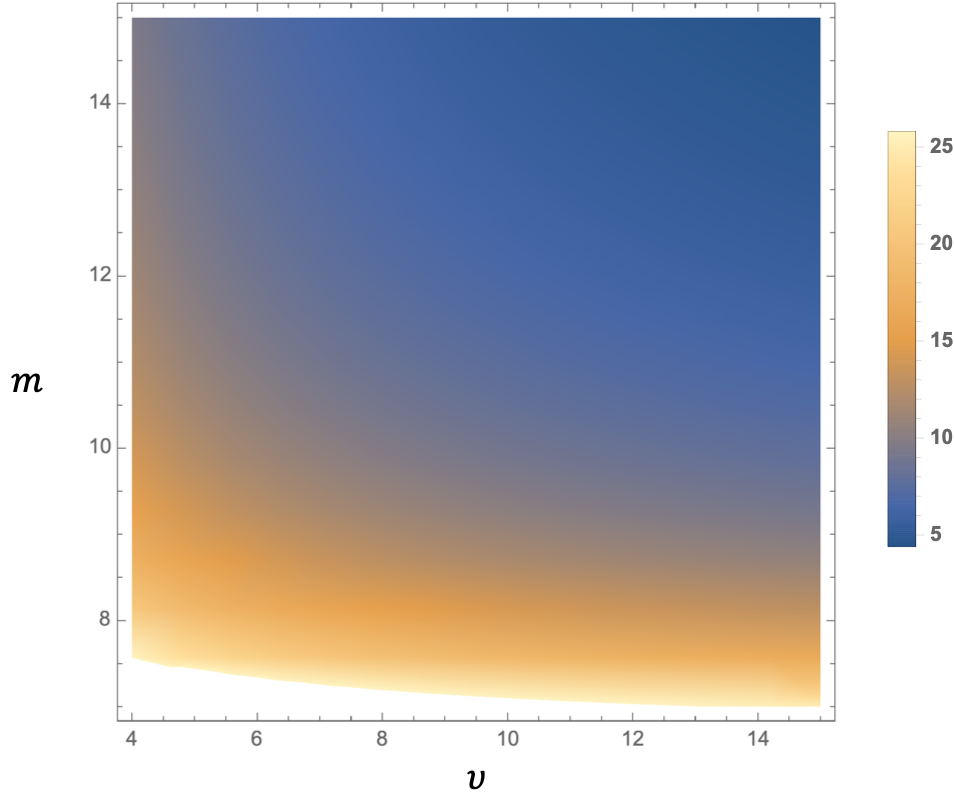}
    \caption{The density plot of $\gamma(m,\nu,\bm{\Sigma})$ when $p=3$ and $\bm{\Sigma}=\bm{I}$.}
    \label{fig:density}
\end{figure}

The spatial term $P(u,v|\beta,\xi)$ is intractable and so we use a Monte Carlo approximation to study the function. In Figure \ref{fig:function}, the function is computed under the scenario that the image is a 1-D grid with $K=100$ and $\xi=0$. The spatial term $P(u,v|\beta,\xi)$ increases with distance and larger spatial parameter $\beta$ leads to the stronger spatial dependence. We also compute the function value with different $K$ in Figure \ref{fig:function_K}. We have $\lim_{|u-v|\rightarrow\infty}P(u,v|\beta,\xi)=1-\frac{1}{K}$, where relevant result can be found in studies of extreme value analysis \citep[see][]{reich2018spatial}. Increasing $K$ leads to smaller spatial dependence. Hence, we fix $K$ to be large to eliminate long-range dependence (i.e., $P(u,v|\beta,\xi)<1$ for large $|u-v|$) and estimate $\beta$ to capture local dependence.

\begin{figure}[ht!]
    \centering
    \begin{subfigure}[ht!]{0.45\textwidth}
    \includegraphics[width=\textwidth]{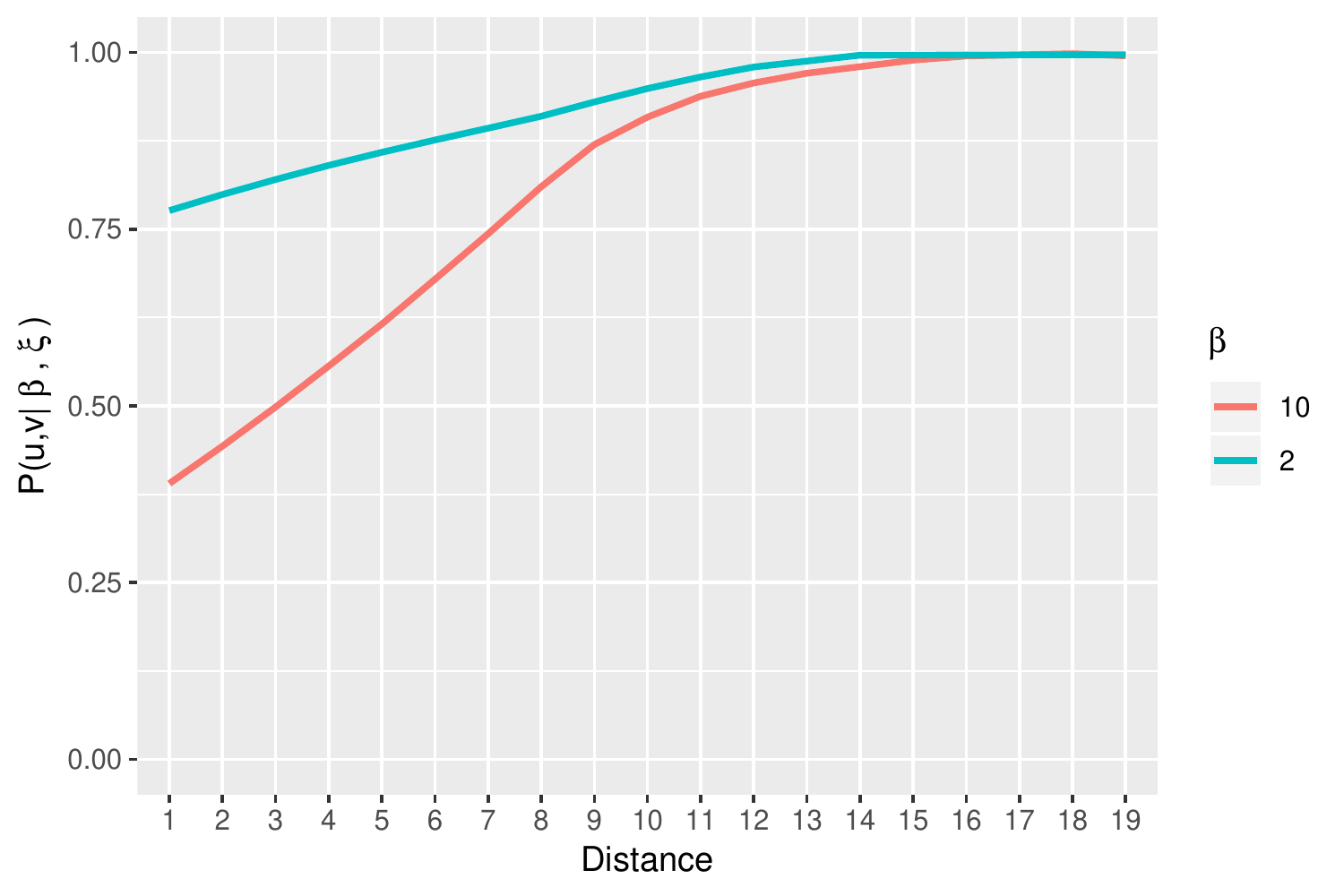}
     \caption{$K=100$}
     \label{fig:function}
     \end{subfigure}
     ~
         \begin{subfigure}[ht!]{0.45\textwidth}
    \includegraphics[width=\textwidth]{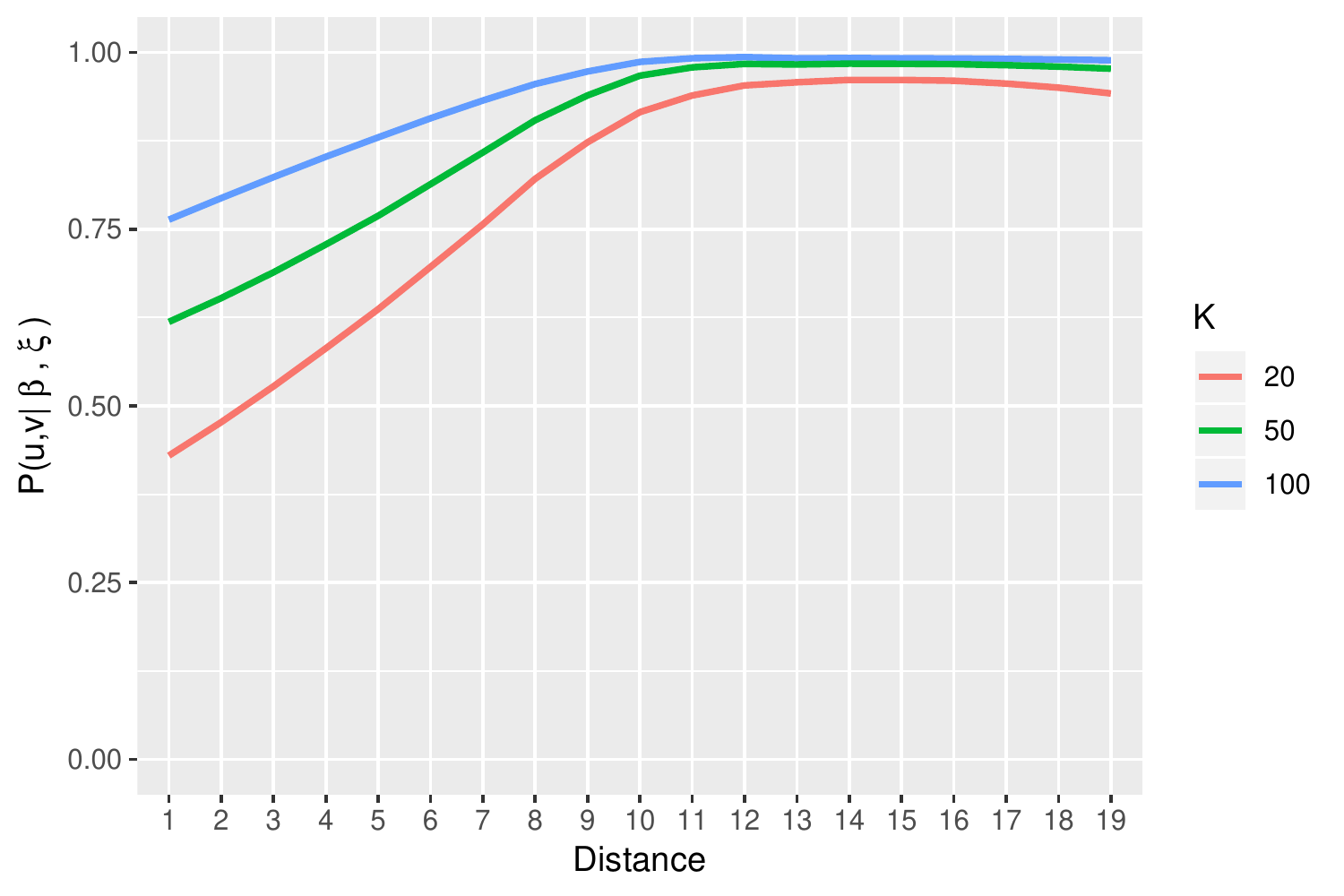}
     \caption{$\beta=10$}
      \label{fig:function_K}
     \end{subfigure}
     \caption{Monte Carlo approximation of the spatial term $P(u,v|\beta,\xi)$. The value varies depending on the distance $|u-v|$, the number of clusters $K$, and the spatial parameter $\beta$; $\xi=0$.}
\end{figure}

For a more intuitive understanding of this model, we also simulate the DTs and visualize the DTs as ellipsoids in a $40\times 40$ grid. In these simulations, we use $\xi=0$, $m=4$, and $\nu=30$. In Figure \ref{fig:sim1}, the DTs within the same latent cluster label are similar to each other, indicating that spatial dependence of the DTs can be achieved by the latent cluster labels following the Potts model. In Figure \ref{fig:sim2}, larger spatial parameter $\beta$ leads to a realization with more dependence on their neighbors. Figure \ref{fig:sim} also illustrates that the Potts model allows sharp breaks, which is desirable if neighboring voxels are in different tracts.

\begin{figure}[ht!]
    \centering
    \begin{subfigure}[ht!]{0.8\textwidth}
       \includegraphics[width=\textwidth]{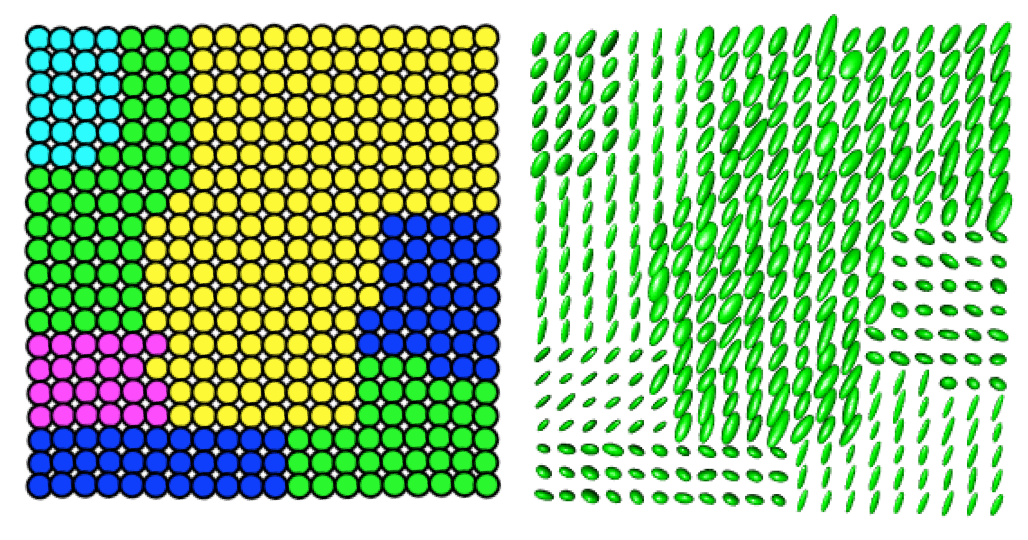}
       \caption{The left panel is the latent cluster labels $g_{v}$ and each color denotes for a distinct latent cluster label; The right panel is the corresponding simulated DTs $\bm{A}_i(\bm{s})$.}\label{fig:sim1}
    \end{subfigure}
     \begin{subfigure}[ht!]{0.8\textwidth}
    \includegraphics[width=\textwidth]{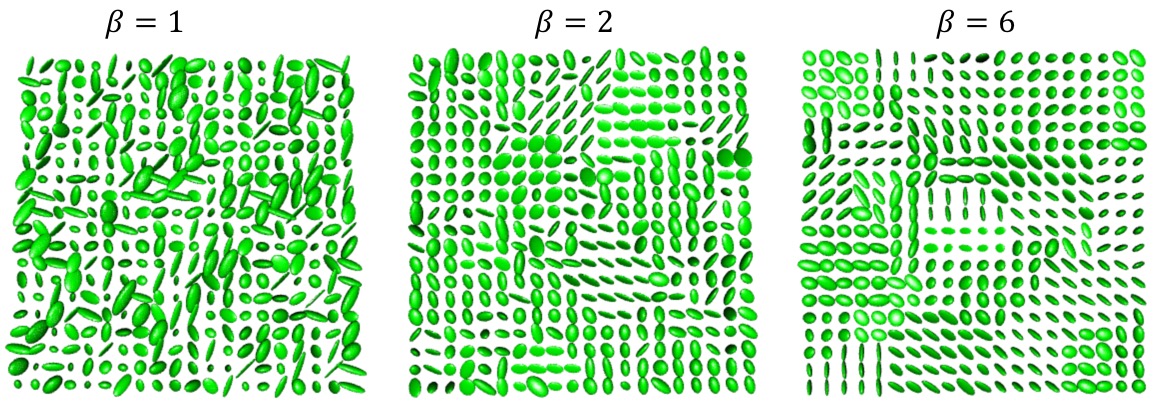}
    \caption{The panels from the left to the right are simulated DTs under the models with $\beta=1,2,6$; $K=20$; $\xi=0$.}\label{fig:sim2}
    \end{subfigure}
    \caption{Simulated DTs based on the proposed model.}\label{fig:sim}
\end{figure}

\subsection{Muti-Subject Model}
\label{sec:multi}
Motivated by the cocaine users data set \citep{ma2017preliminary} that includes 11 cocaine users and 11 non-cocaine users, we extend the single-subject model to the multi-subject setting. The clinical objective is to analyze the brain's physical structure for differences between the two groups. The objective can be statistically formulated as finding regions in the brain where the distribution of the DTs across the subjects is different between cocaine users and non-cocaine users.

Let $\bm{A}_{iv}$ be the DT and $g_{iv}$ be the cluster label for voxel $v\in\{1,2,...,n\}$ and subject $i\in\{1,2,...,N\}$. By extending $g_v$ to $g_{iv}$, the subject-level cluster labels not only model intra-subject spatial dependence but also allow inter-subject variability. As in the single-subject model, the DTs are conditionally independent given the random matrices $\bm{M}_{iv}$ following a finite mixture model:
\begin{equation}
    \bm{A}_{iv}|\bm{M}_{iv}, m\distas{indep.} \mathcal{IW}_p(\bm{M}_{iv},m),\ \bm{M}_{iv}:= \bm{V}_{g_{iv}},\ \bm{V}_k\distas{i.i.d} \mathcal{W}_p(\bm{\Sigma},\nu).
\end{equation}
However, the latent Potts model is generalized to account for multiple subjects. We define $x_i$ as the binary group indicator of subject $i$. In the motivating data, cocaine users have $x_i=1$ and non-cocaine users have $x_i=0$. To model intra-subject spatial dependence within a group, we extend the latent cluster process by introducing the group-level cluster labels $h_{x v}$ for group $x\in\{0,1\}$ and voxel $v\in\{1,2,...,n\}$. Both $h_{x v}$ and $g_{iv}$ are also spatially dependent with full conditional distributions
\begin{equation}
\label{eq:cond}
\small
    \begin{aligned}
    P(g_{iv}=k|\alpha,\beta,\xi,\bm{g}_{-iv},\bm{h})&\propto \exp\left[-k^{\xi}+\beta\sum_{u\in N_v}\mathcal{I}(g_{iu}=k)+\alpha\mathcal{I}(h_{x_iv}=k)\right]\\
    P(h_{xv}=k|\alpha,\beta,\bm{h}_{-xv},\bm{g})&\propto \exp\left[\beta\sum_{u\in N_v}\mathcal{I}(h_{xu}=k)+\sum_{j:x_{j}=x}\alpha\mathcal{I}(g_{jv}=k)\right],
    \end{aligned}
\end{equation}
where $\bm{g}_{-(iv)}$ is the set on $\bm{g}_i=\{g_{i1}, ..., g_{in}\}$ excluding $g_{iv}$, $\bm{h}_{-(xv)}$ is the set $\bm{h}_x=\{h_{x1}, ..., h_{xn}\}$ excluding ${h}_{xv}$, $\bm{g}$ is the set on $\{\bm{g}_1, ..., \bm{g}_N\}$, and $\bm{h}$ is the set on $\{\bm{h}_0,\bm{h}_1\}$. The joint probability mass function (PMF) of $\{\bm{g}_1, \bm{g}_2, ..., \bm{g}_N\}\cup\{\bm{h}_0, \bm{h}_1\}$ is given in Appendix C. Since the conditional densities in (\ref{eq:cond}) satisfy the conditions of the Hammersley-Clifford Theorem \citep{clifford1990markov}, the existence of joint distribution of $\{\bm{g}_1, \bm{g}_2, ..., \bm{g}_N\}\cup\{\bm{h}_0, \bm{h}_1\}$ is guaranteed (Appendix C).

A graphical representation of this latent Potts model is provided in Figure \ref{fig:random_field}. Cluster labels $\bm{g}_i$ and $\bm{h}_x$ can be understood as the spatial pattern of subject $i$ and the general spatial pattern of subjects in group $x$, respectively. In comparison to the single-subject model, the group-clustering parameter $\alpha$ is introduced for modeling multiple subjects. If $\alpha=0$, $\bm{A}_{iv}$ is independently distributed over subjects; otherwise, the subject-level cluster label $g_{iv}$ depends on the group-level cluster label $h_{x_iv}$, leading to the smaller inter-subject variability of spatial dependence pattern within one group. 

\begin{figure}[ht!]
    \centering
    \includegraphics[width=1\textwidth]{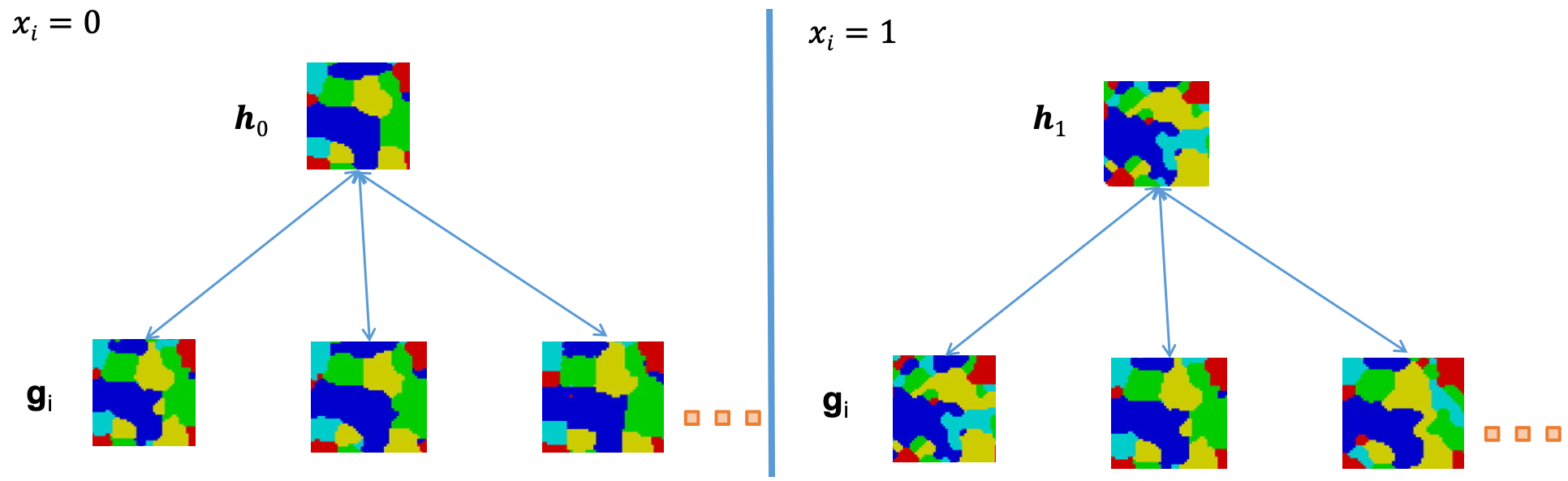}
    \caption{The graphical representation of the latent cluster labels. The group-level cluster labels are $\bm{h}_x=\{h_{x1}, ..., h_{xn}\}$ and the subject-level cluster labels are $\bm{g}_i=\{g_{i1}, ..., g_{in}\}$. Cluster labels $\bm{h}_x$ and $\bm{g}_i$ are mutually dependent. The subject-level cluster labels $\bm{g}_i$ have inter-subject variability. The group-level cluster labels $\bm{h}_x$ are a summary of the spatial dependence of all subjects.}
    \label{fig:random_field}
\end{figure}

 To further understand the role of $\alpha$ and $h_{xv}$, we inspect the density of $\bm{A}_{iv}$ conditioned on $h_{xv}$ and marginal over all other labels (Appendix D). The conditional density of $\bm{A}_{iv}$ given $x_i=x$ and $h_{xv}$ is the mixture of inverse Wishart distributions proportional to 
\begin{equation}
\label{eq:cluster}
    \sum_{k=1}^K \exp\left[-k^{\xi} + \alpha\mathcal{I}(h_{xv}=k)\right] \mathcal{IW}_p(\bm{A}_{iv}|\bm{V}_k,m),
\end{equation}
 where the term $\exp\left[-k^{\xi} + \alpha\mathcal{I}(h_{xv}=k)\right]$ is the proportional weight of cluster $k$. Therefore, $h_{xv}=k$ elevates the mass on mixture component $k$ at voxel $v$ for all subjects with $x_i=x$. Assuming $\alpha>0$, the conditional density (\ref{eq:cluster}) depends on $x_i$ if and only if $h_{0v}\ne h_{1v}$.
 
 The clinical objective is to find regions of differences between two groups, which can be formulated into finding voxels for which the distribution of $\bm{A}_{iv}$ is different for $x_i=0$ or $x_i=1$. As shown in the conditional density (\ref{eq:cluster}), the test can be simplified to the test that
      \begin{equation}
      \label{eq:hypo}
            \begin{aligned}
            &\mathcal{H}_{ov}: h_{0v}=h_{1v}\\
            &\mathcal{H}_{av}: h_{0v}\not=h_{1v}.
            \end{aligned}
        \end{equation}
Bayesian inference provides estimates of the posterior probabilities of the hypotheses, which is further discussed in Section \ref{sec:computing}.

We investigate spatial dependence as in the single-subject model. To measure spatial dependence within and across subjects, we propose the variogram \begin{equation}
\label{eq:multi_vario}
    \mathcal{V}_{ij}(u,v)=\mathbb{E}||\bm{A}_{iu}-\bm{A}_{jv}||_F^2=\gamma(m,\nu,\bm{\Sigma})P_{ij}(u,v|\alpha,\beta,\xi)
\end{equation} with $i=j$ for individual variogram and $i\not=j$ for inter-subject variogram, respectively. For inter-subject variogram, we also compare the within-group variogram for subjects with $x_i=x_j$ and between-group variogram for subjects with $x_i\ne x_j$. Both individual variogram and inter-subject variogram are also separable \citep{cressie1992statistics}. $\gamma(m,\nu,\bm{\Sigma})$ is the non-spatial term which has been discussed in Section \ref{sec:single}. The spatial term $P_{ij}(u,v|\alpha,\beta,\xi)$ is the marginal probability of $g_{iu}\not=g_{jv}$.

In Figure \ref{fig:function_multi}, the function $P_{ii}(u,v|\alpha,\beta,\xi)$ is computed using a Monte Carlo approximation under the scenario that the image is a 1-D grid with $N=5$ and $K=100$. The spatial parameter $\beta$ largely controls the within-subject dependence. In Figure \ref{fig:within_function_multi} plotting $P_{ij}(u,v|\alpha,\beta,\xi)$, larger $\alpha$ leads to more dependence in the within-group variogram. Therefore, $\alpha$ controls the dependence of subjects within one group. Since $g_{iu}$ and $g_{ju}$ are assumed to be independent, the spatial term $P_{ij}(u,v|\alpha,\beta,\xi)$ is a constant $1-\frac{1}{K}$ in between-group variogram (Figure \ref{fig:cross_function_multi}). 
\begin{figure}[ht!]
    \centering
    \begin{subfigure}[ht!]{0.3\textwidth}
        \includegraphics[width=\textwidth]{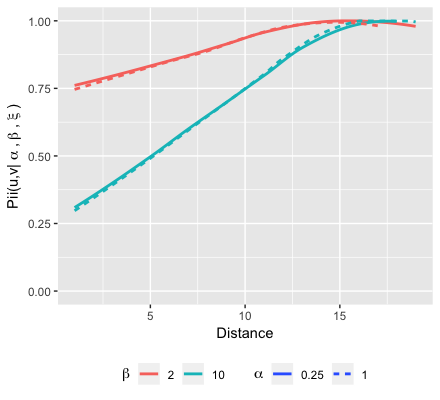}
        \caption{Individual variogram}
        \label{fig:function_multi}
    \end{subfigure}
      \begin{subfigure}[ht!]{0.3\textwidth}
        \includegraphics[width=\textwidth]{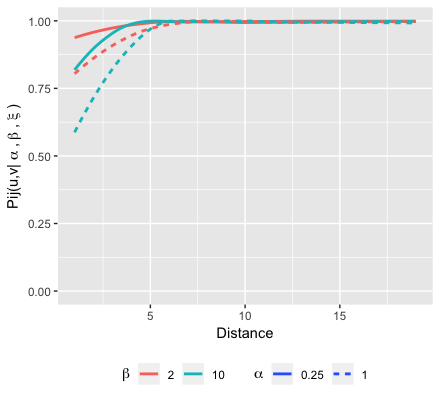}
        \caption{Within-group variogram}
        \label{fig:within_function_multi}
    \end{subfigure}
      \begin{subfigure}[ht!]{0.3\textwidth}
        \includegraphics[width=\textwidth]{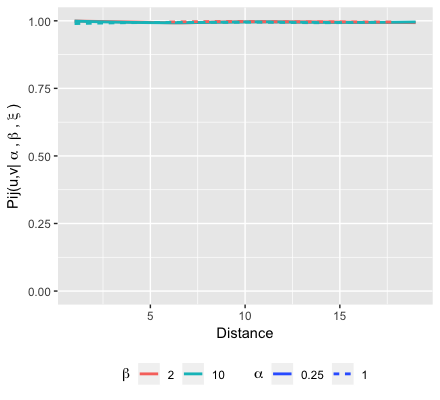}
        \caption{Between-group variogram}
        \label{fig:cross_function_multi}
    \end{subfigure}
    \caption{Monte Carlo approximation of the spatial term $P_{ij}(u,v|\alpha,\beta,\xi)$. The value varies depending on the distance $|u-v|$, the group-clustering parameter $\alpha$, and the spatial parameter $\beta$; $\xi=0$.}\label{fig:theoretical_vario}
\end{figure}

\section{Computation}
\label{sec:computing}
We use MCMC to fit the model described in Section \ref{sec:model} (Appendix E). The codes are written in hybrid \proglang{R} and \proglang{C++} codes. The final model is
\begin{equation}
    \begin{aligned}
    &\bm{A}_{iv}|\bm{M}_{iv},m\distas{indep.} \mathcal{IW}_p(\bm{M}_{iv},m),\ \bm{M}_{iv}:= \bm{V}_{g_{iv}},\ \bm{V}_k\distas{i.i.d} \mathcal{W}_p(\bm{\Sigma},\nu)\\
    &P(g_{iv}=k|\alpha,\beta,\xi,\bm{g}_{-(iv)},\bm{h})\propto \exp\left[-k^{\xi}+\beta\sum_{u\in N_v}\mathcal{I}(g_{iu}=k)+\alpha\mathcal{I}(h_{x_iv}=k)\right]\\
    & P(h_{xv}=k|\alpha,\beta,\bm{h}_{-(xv)},\bm{g})\propto \exp\left[\beta\sum_{u\in N_v}\mathcal{I}(h_{xu}=k)+\sum_{j:x_{j}=x}\alpha\mathcal{I}(g_{jv}=k)\right]
    \end{aligned}
\end{equation}
which is referred to as the \textit{Potts Model} in the rest of the paper. Using the moment method \citep{robert2007bayesian}[Section 3.2.4], we set $\bm{\Sigma}$ as the sample mean of all observed DTs. The priori information brought by $\bm{\Sigma}$ has a little impact if the number of observations is large. We put uniform prior for the degrees of freedom $m$ and $\nu$ on $[5,50]\times [4,50]$. Following \citet{liang2010double}, we put a uniform prior for $\bm{\theta}=\{\alpha,\beta,\xi\}$ on $[0,20]\times [0,20] \times [0, 1]$, denoted as $\pi(\bm{\theta})$. Below we describe the updating rule for each parameter.

The MCMC algorithm is a combination of Gibbs and Metropolis-Hastings steps. The latent mean matrices and cluster labels are updated via Gibbs steps. Their full conditional distributions are 
\begin{itemize}
\small
    \item $\bm{V}_k|.\sim\mathcal{W}_p( (\bm{\Sigma}^{-1}\nu+(m-p-1)\sum_i\sum_{v:g_{iv}=k}\bm{A}_{iv}^{-1})^{-1}(N n_{k}m+\nu),N n_{k}m+\nu)$ \\
    \item $P(g_{iv}=k|.)\propto \mathcal{IW}_p(\bm{A}_{iv}|\bm{V}_{g_{iv}},m)\exp\left[-k^{\xi}+\beta\sum_{u\in N_v}\mathcal{I}(g_{iu}=k)+\alpha\mathcal{I}(h_{x_i, v}=k)\right]$
    \item $P(h_{xv}=k|.)\propto \exp\left[\beta\sum_{u\in N_v}\mathcal{I}(h_{xu}=k)+\alpha\sum_{j:x_{j}=x}\mathcal{I}(g_{jv}=k)\right]$
\end{itemize}
where $n_k=\sum_{i,v}\mathcal{I}(g_{iv}=k)$. In addition, $P(g_{iv}=k|.)$ and $P(h_{xv}=k|.)$ can be updated in parallel over $i$ and $x$, respectively. Since the uniform prior is not conjugate, we have to sample $[m|.]$ and $[\nu|.]$ via Metropolis-Hastings sampling with log-normal random walk as proposal distribution. 

To select regions of differences via Bayesian hypothesis testing, we reject the null hypothesis in (\ref{eq:hypo}) if $P(h_{0v}\not=h_{1v}|.)<P(h_{0v}=h_{1v}|.)$. The posterior probabilities can be estimated through MCMC samples that $h_{0v}=h_{1v}$ or $h_{0v}\not=h_{1v}$.

Updating the Potts hyperparameters $\alpha$, $\beta$, and $\xi$ is problematic because the normalizing constant in the joint distribution function of the cluster labels is intractable (see the joint PMF in Appendix C). A simple approach is to estimate the parameters outside of MCMC. The plug-in values can be obtained from cross-validation  \citep[i.e.,][]{goldsmith2014smooth}, pseudo-likelihood comparison \citep[i.e.,][]{zhao2014bayesian,lan2016bayesian}, or by comparing empirical and model-based variograms (i.e., Figure \ref{fig:empi} and \ref{fig:MCMC}). However, these methods fail to account for uncertainty about these imputed parameters and so we update them using the double Metropolis-Hastings algorithm \citep{liang2010double}.

\citet{park2018bayesian} review several Monte Carlo methods for models with intractable normalizing constants and recommend the double Metropolis-Hastings algorithm proposed by \citet{liang2010double} because of its ease of implementation and computational efficiency. \citet{li2018bayesian} combine the double Metropolis–Hastings algorithm with usual Bayesian tools for implementing the Potts model. The double Metropolis-Hastings update for $\bm{\theta}$ begins with a candidate $\bm{\theta}'$ drawn from $q(\bm{\theta}|\widetilde{\bm{\theta}})$ where $\widetilde{\bm{\theta}}$ is the current value and $q(\bm{\theta}|\widetilde{\bm{\theta}})$ is a log-normal random walk transitional probability centered at $\widetilde{\bm{\theta}}$. Given the candidate $\bm{\theta}'$, we draw labels $\bm{g}_i'=\{g_{iv}',...,g_{in}'\}$ and $\bm{h}_x'=\{h_{xv}', ..., h_{xv}'\}$ using Gibbs sampling for each $i$ and $x$, respectively. The candidate $\bm{\theta}'$ is accepted with the probability $\min (1, r)$ where $r=\frac{\pi(\bm{\theta}')\mathcal{P}(\bm{g}',\bm{h}'|\widetilde{\bm{\theta}})\mathcal{P}(\widetilde{\bm{g}},\widetilde{\bm{h}}|{\bm{\theta}'})}{\pi(\widetilde{\bm{\theta})}\mathcal{P}(\widetilde{\bm{g}},\widetilde{\bm{h}}|\widetilde{\bm{\theta}})\mathcal{P}(\bm{g}',\bm{h}'|{\bm{\theta}'})}$ where $\mathcal{P}(\bm{g},\bm{h}|{\bm{\theta}})$ is the likelihood of $\{\bm{g}_1, \bm{g}_2, ..., \bm{g}_N\}\cup\{\bm{h}_0, \bm{h}_1\}$ conditioned on $\bm{\theta}$. $\widetilde{\bm{g}}_{i}$ and $\widetilde{\bm{h}}_{x}$ are current values. Due to the concern that the double Metropolis-Hastings algorithm is not an exact sampling \citep{liang2010double,park2018bayesian}, the MCMC convergence of $\bm{\theta}$ in the simulation studies (Section \ref{sec:sim}) and the real data analysis (Section \ref{sec:application}) are monitored by Heidelberger and Welch's convergence diagnostic \citep{heidelberger1981spectral}.

\section{Simulation}
\label{sec:sim}
In this section, we illustrate the performance of our method using two simulation studies under different scenarios for synthetic data. We compare our method to the non-spatial DTI inference method the Gaussian symmetric matrix model \citep{schwartzman2008inference} referred to as the \textit{Random Ellipsoid Model}. The \textit{Random Ellipsoid Model} is a non-spatial matrix-variate method and assumes that $\bm{A}_{iv}$ follows a Gaussian symmetric random matrix distribution with the probability density function (PDF) as $f(\bm{A}_{iv}|x_i=x,\bm{\Sigma}_v,\sigma^2)=H(\bm{A}_{iv})\exp\left[\frac{1}{2\sigma^2}Tr(2\bm{\Sigma}_{xv}\bm{A}_{iv}-\bm{\Sigma}_{xv}\bm{\Sigma}_{xv})\right]$, where $\bm{\Sigma}_{xv}$ is the DT's population mean at voxel $v$ of group $x$ and $\sigma^2$ is the nuisance parameter. The \textit{Random Ellipsoid Model} selects regions of differences via testing $\bm{\Sigma}_{0v}=\bm{\Sigma}_{1v}$, where the test statistics are constructed by maximum likelihood estimations. the \textit{Potts Model} has $8,000$ MCMC samples with $3,000$ discarded as burn-in. Methods are evaluated in terms of true positive rate (TPR), false positive rate (FPR), false discovery rate (FDR), and typical computation time. 

We first investigate the performance of our method when the data are generated from a mixture model. We use a $40\times 40$ grid with spacing 1 between adjacent grid points as an image. Each simulated data set consists of 5 subjects in the control group ($x_i=0$) and 5 subjects in the treatment group ($x_i=1$). For the control group ($x_i=0$), we equally partition the graph into 4 parts by rectangular regions so that $g_{iv}\in\{1,2,3,4\}$, ordered by right-to-left. Thus each region is a $40\times 10$ region. The treatment group has the same partition as the control group, except a $10\times10$ region at the middle of the second region where $g_{iv}=5$. This simulates the brain with a small region of difference between the two groups. For each simulation, $\bm{\Sigma}_k$ is generated based on the model $$\bm{\Sigma}_k\sim\mathcal{W}_3((k+1)\bm{I}_3, 30)\text{ for $k=1,2,3,4$};\ \bm{\Sigma}_5\sim\mathcal{W}_3(1.5\bm{I}_3, 30).$$ Given the simulated $\bm{\Sigma}_k$, the data are generated based on the model $\bm{A}_{iv}|\bm{\Sigma}_{g_{iv}}\sim \mathcal{IW}_3(\bm{\Sigma}_{g_{iv}},5)$. 

Our model with $K=10, 50, 100$ is compared to the \textit{Random Ellipsoid Model}. The results averaged over 50 data sets are summarized in Table \ref{t:sim1}. Our model has significantly improved performance in terms of the TPR, FPR, FDR in comparison to the \textit{Random Ellipsoid Model}. The small number of subjects might be one of the causes that the alternative produces low TPR. \citet{schwartzman2008inference} discusses that the accuracy of maximum likelihood estimates of the \textit{Random Ellipsoid Model} is dependent on the number of subjects. Since the choice of $K$ does not affect selection accuracy, the simulation results also support the claim that the model can be less sensitive to the number of clusters if $K$ is larger than the true clusters. 

To determine robustness to model misspecification, we also simulate data from the spatial Cholesky process described as follows: The DT matrix for subject $i$ at voxel $v$ is determined by six independent spatial Gaussian processes $U_{ivk}$ ($k\in\{1,2,...,6\}$). These spatial Gaussian processes are arranged in the lower triangular matrix $\bm{L}_{iv}$ with $\bm{L}_{iv}=\begin{bmatrix}e^{U_{iv1}} &0 & 0\\
U_{iv4} & e^{U_{iv2}} &0\\
U_{iv5} & U_{iv6} & e^{U_{iv3}}\end{bmatrix}$. The responses $\bm{A}_{iv}$ are then constructed as $\bm{A}_{iv} = \bm{L}_{iv}\bm{L}_{iv}^T$, thereby introducing spatial dependence and guaranteeing positive definite of responses. We again use a $40\times 40$ grid with spacing 1 between adjacent grid points as an image. There are 10 subjects in the control group and 10 subjects in the treatment group. The six spatial Gaussian processes are simulated with variance $\tau^2=0.1$ and exponential correlation function with range parameter $\rho=2$. The mean of the six Gaussian processes are all $0$ except for treatment subjects' $10\times10$ region in the center of the image where $U_{ivk}$ has mean $0.5$ for $k\leq3$ and $0.25$ for $k>3$. This simulates the brain with a small region of difference between the two groups. We compare our model with $K=10, 50, 100$ to the \textit{Random Ellipsoid Model}. The results averaged over 50 simulations are summarized in Table \ref{t:sim2}. The results demonstrate that our model maintains good performance, indicating the \textit{Potts Model} is robust to this form of misspecification. In addition, under the spatial dependence assumption, the spatial models produce an overall better performance than the non-spatial model.

A problematic issue to the use of Bayesian methods in neuroimaging data is their heavy computational burden. In both simulations, the \textit{Potts Model} has a computational speed within a few hours. The \textit{Random Ellipsoid Model} avoids the expensive MCMC, but since the performance of the \textit{Random Ellipsoid Model} is too conservative, the \textit{Potts Model} is a reasonable trade-off.

\begin{center}
\begin{table}[ht!]
\caption{The simulation results. The true positive rate, false positive rate, false discovery rate, and typical computation time of the \textit{Potts Model} and \textit{Random Ellipsoid Model} are summarized.}

\begin{subtable}[!ht]{\textwidth}\centering
\caption{Data generated from mixture models}\label{t:sim1}
\begin{tabular}{ccccc}
\hline\hline
\multirow{2}{*}{Method} & \multicolumn{3}{c}{\textit{Potts}}  &\multirow{2}{*}{\textit{\begin{tabular}[c]{@{}c@{}}Random \\ Ellipsoid\end{tabular}}} \\
 & K=10 & K=50 & K=100 &  \\ \hline
TPR & 0.99 & 0.97 & 0.98 &0.29\\
FPR & 0.013 & 0.010 & 0.009  &0.00\\
FDR & 0.025 & 0.025 & 0.024  &0.002\\ \hline
Time (hours) & 0.5 & 0.8 & 1.0 &$<$0.01\\ \hline\hline\\
\end{tabular}
\end{subtable}

\begin{subtable}[!ht]{\textwidth}\centering
\caption{Data generated from the spatial Cholesky process }\label{t:sim2}
\begin{tabular}{cccccc}
\hline\hline
\multirow{2}{*}{Method} & \multicolumn{3}{c}{\textit{Potts}}  & \multirow{2}{*}{\textit{\begin{tabular}[c]{@{}c@{}}Random \\ Ellipsoid\end{tabular}}} \\
 & K=10 & K=50 & K=100 &  &  \\ \hline
TPR & 0.79 & 0.80 & 0.77 &  0.50 \\
FPR & 0.01 & 0.01 & 0.01  & 0.00 \\
FDR & 0.03 & 0.03 & 0.02  & 0.03 \\\hline
Time (hours) & 1.0 & 1.5 & 1.8  &$<$0.01\\\hline\hline
\end{tabular}
\end{subtable}
\end{table}
\end{center}

\section{Real Data Application}
\label{sec:application}
We apply this model to the data set of cocaine users  \citep{ma2017preliminary} described in Section \ref{sec:model}. The data are provided by the Institute for Drug and Alcohol Studies of Virginia Commonwealth University (VCU). The study recruited 11 cocaine users and 11 controls to test for microstructural changes of the brain through DTI. In the data analysis, we focus on the corpus callosum containing 15,273 voxels because this region plays important roles such as transferring motor, sensory, and cognitive information between the brain hemispheres \citep{ma2009diffusion}. Conventionally, studies on cocaine use focus on this region \citep[i.e.,][]{ma2009diffusion,ma2017preliminary,lane2010diffusion}. 

Before model fitting, we examine the fit of the proposed model to the cocaine users data via empirical estimates of variograms. We denote $\hat{\mathcal{V}}_{ij}(d)=\frac{1}{N_d}\sum_{|u-v|=d}||\bm{A}_{iu}-\bm{A}_{jv}||_F^2$ as the empirical variogram value of subjects $i$ and $j$ at distance $d$, where $N_d$ is the number of pairs with $|u-v|=d$. We plot these empirical variograms of the motivating data in Figure \ref{fig:empi}. The DTs have a strong within-subject spatial dependence (Figure \ref{fig:emp_function_multi}). The empirical within-group variogram (Figure \ref{fig:emp_within_function_multi}) also increases with distance, indicating inter-subject dependence within a group, however, the empirical variogram is almost flat in the between-group variogram (Figure \ref{fig:emp_cross_function_multi}), which suggests that the subjects are independent if they are in different groups. Since these empirical variograms perfectly match the theoretical variograms in Figure \ref{fig:theoretical_vario}, the the hierarchical Potts  model assumptions about spatial dependence are reasonable. 
\begin{figure}[ht!]
    \centering
    \begin{subfigure}[ht!]{0.3\textwidth}
        \includegraphics[width=\textwidth]{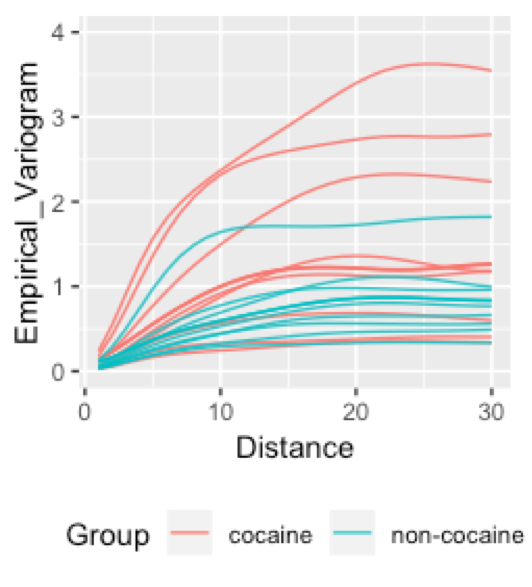}
        \caption{Individual variograms}
        \label{fig:emp_function_multi}
    \end{subfigure}
    ~ 
      \begin{subfigure}[ht!]{0.3\textwidth}
        \includegraphics[width=\textwidth]{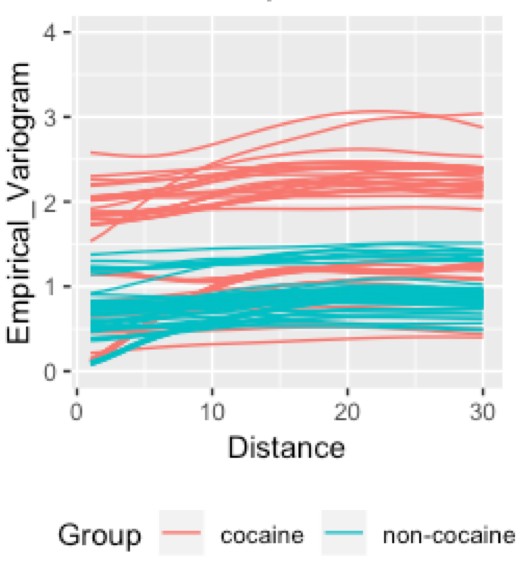}
        \caption{Within-group variograms}
        \label{fig:emp_within_function_multi}
    \end{subfigure}
      ~ 
      \begin{subfigure}[ht!]{0.3\textwidth}
        \includegraphics[width=\textwidth]{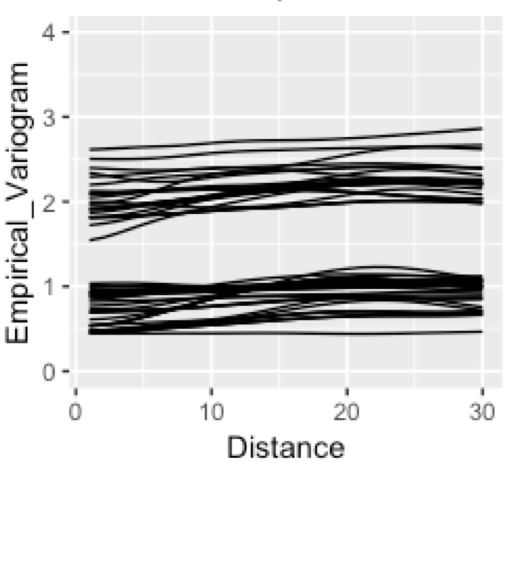}
        \caption{Between-group variograms}
        \label{fig:emp_cross_function_multi}
    \end{subfigure}
    \caption{The empirical variograms of the cocaine users data \citep{ma2017preliminary}. The lines are the empirical variograms for each subject/pair.\label{fig:empi}}
\end{figure}

We sample $11,000$ MCMC samples with $3,000$ discarded as burn-in and it typically takes $5$ hours using a CPU with 3.4 GHz Intel Core i5. To study the sensitivity to $K$, we fit the model with $K$ as $100$, $200$, $300$, $400$ and $500$ and use the Rand index \citep{rand1971objective} for measuring the similarities of regions of differences detection with different $K$. The Rand index measures clustering similarity: if the two clusterings are almost identical, the index is close 1; otherwise, the index is close 0. In Table \ref{tab:rand}, the Rand indices of any two $K$ are close to 1, hence the selection is not sensitive to $K$. For a concise illustration, we use the result of $K=100$ in the rest of this section. 
\begin{table}[ht]
\caption{The Rand index for measuring clustering similarities. The off-diagonals of the table are the Rand indices for any two $K$.}
\label{tab:rand}
\begin{center}
\begin{tabular}{c|ccccc}
\hline\hline
K & 100 & 200 & 300 & 400 & 500 \\ \hline
100 & . &0.92 & 0.91 & 0.92 & 0.92\\
200 & . &  . & 0.94 & 0.95 & 0.94\\
300 & . & . & . & 0.96 & 0.95\\ 
400 & . & . & . & . & 0.96\\ 
500 & . & . & . & . & .\\ 
\hline\hline
\end{tabular}
\end{center}
\end{table}

We first use the \proglang{R} package \code{brainR} \citep{muschelli2014brainr} for 3-D visualizing the regions of differences. To investigate if the performance is improved by introducing spatial dependence, we also compare it to the \textit{Random Ellipsoid Model}. We give the confidence level $0.9$ for the \textit{Random Ellipsoid Model}. The selected regions of differences are displayed in Figure \ref{fig:ind_Armin}. The Rand index of the two clusterings is $0.86$ and so the results of the two analyses are similar but the \textit{Potts Model} finds more spatial contiguous regions. Our study shows that a region of difference is detected in the splenium, which is consistent with previous clinical studies on cocaine use \citep[see][]{lane2010diffusion}. The splenium is a component located on the posterior end of the corpus callosum with an essential role on cognition. Since many studies revealed that the disease status at proximally-located/neighboring voxels can be similar \citep[see][]{wu2013mapping,xue2018bayesian}, our method may be more clinically meaningful than other studies because the detected regions are spatially contiguous. Furthermore, our test is able to find more regions of differences compared to alternative clinical studies on investigating the effect of cocaine use. For example, \citet{ma2017preliminary} did not find such regions of differences by using the same data set. 

\begin{figure}[ht!]
    \centering
    \includegraphics[width=0.45\textwidth]{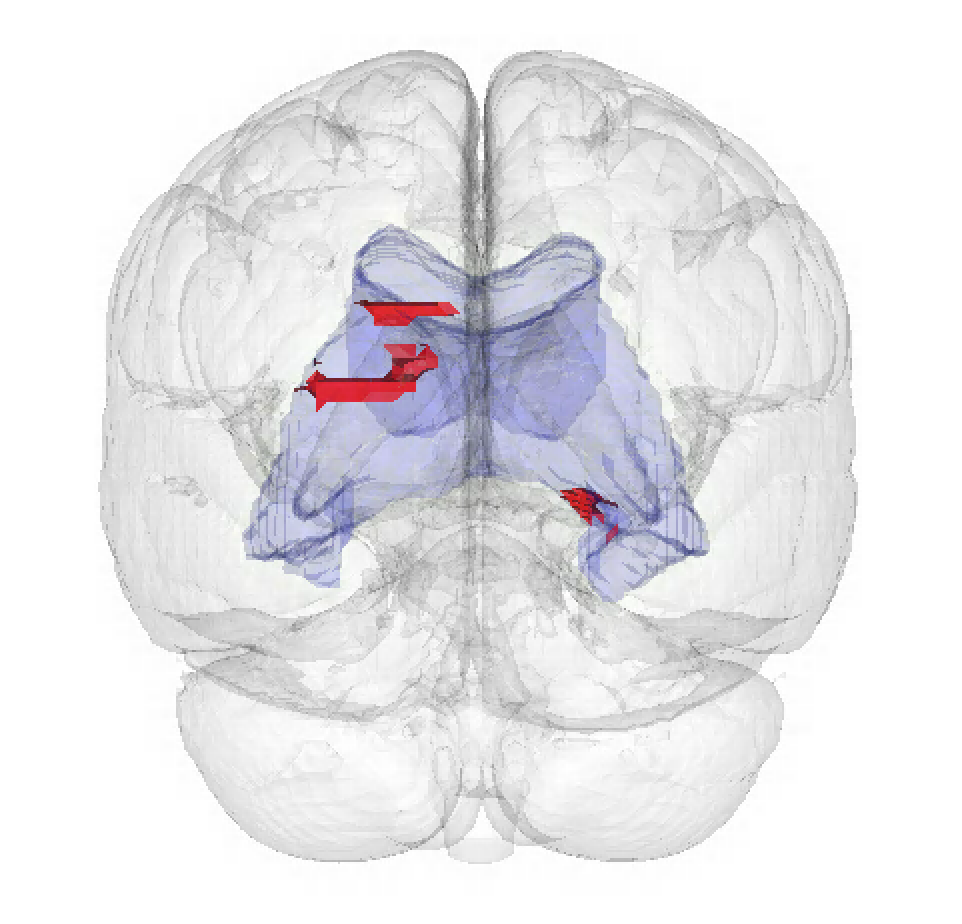}
    \includegraphics[width=0.45\textwidth]{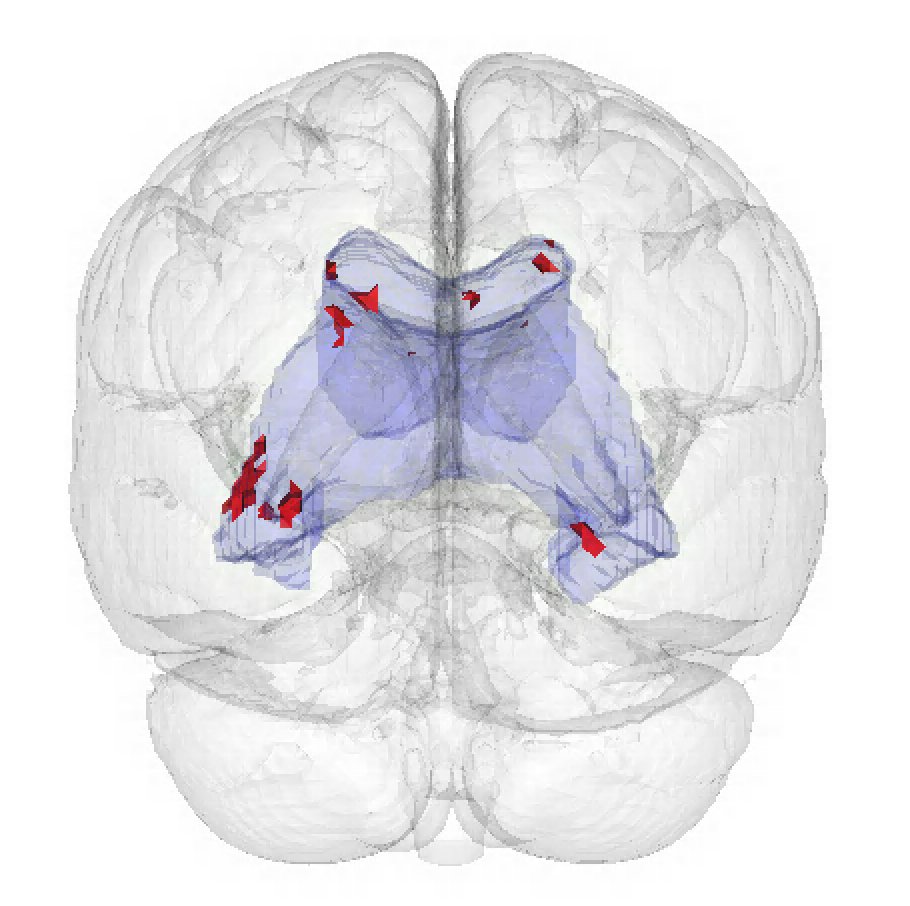}
    \caption{The regions of differences between cocaine users and non-cocaine users. This is the regions of differences selected by the \textit{Potts Model} (left panel) and the \textit{Random Ellipsoid Model} (right panel). The red area is the regions of differences.}
    \label{fig:ind_Armin}
\end{figure}

The MCMC trace plots and posterior densities of the Potts spatial dependence parameters $\bm{\theta}$ are in Figure \ref{fig:MCMC}. The concentration parameter $\xi$ has $95\%$ credible region $[0.884,0.888]$, indicating there are a few active clusters. The group-clustering parameter $\alpha$ and spatial parameter $\beta$ control the within and inter subject spatial dependence and have $95\%$ credible regions $[0.323,0.327]$ and $[18.698,18.703]$, respectively. The dependence information revealed by the two credible regions is identical to the information obtained from the empirical variograms (Figure \ref{fig:empi}). Thus similar to the usage of the classic variogram, the generalized empirical variograms may also be a tool for obtaining the plug-in values of hyperparameters as an alternative \citep[see][]{reich2018spatial}.

\begin{figure}[ht!]
    \centering
    \begin{subfigure}[ht!]{0.5\textwidth}
        \includegraphics[width=\textwidth]{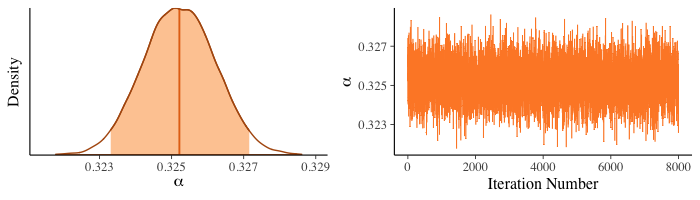}
        \caption{Group-clustering parameter, $\alpha$}
        \label{fig:emp_function_multi}
    \end{subfigure}
    ~ 
      \begin{subfigure}[ht!]{0.5\textwidth}
        \includegraphics[width=\textwidth]{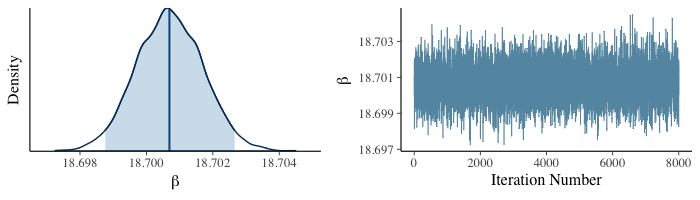}
        \caption{Spatial parameter, $\beta$}
        \label{fig:emp_within_function_multi}
    \end{subfigure}
      ~ 
      \begin{subfigure}[ht!]{0.5\textwidth}
        \includegraphics[width=\textwidth]{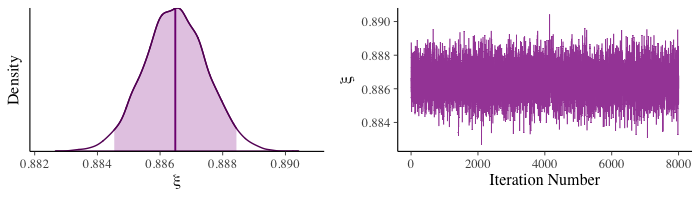}
        \caption{Concentration parameter, $\xi$}
        \label{fig:emp_cross_function_multi}
    \end{subfigure}
    \caption{The MCMC summaries of the hyperparameters $\bm{\theta}$. The left panel is the histogram of posterior samples and the colored region is the $95\%$ credible region. The right panel is the MCMC trace plot of posterior samples.\label{fig:MCMC}}
\end{figure}


\section{Discussion}
\label{sec:dis}
Although the spatial statistics literature on models and tools for matrix-variate data is sparse, the usage of positive definite matrix-variate is broad, which includes multiple-input and multiple-output (MIMO) systems \citep{smith2007distribution} and computer vision \citep{cherian2016bayesian}. Our major contribution is to present a spatial statistics formulation to model matrix-variate data via a Bayesian semiparametric mixture model. Our formulation retains the original data structure and accounts for spatial dependence by a computationally elegant model. In simulation studies, our model produces significantly improved performance compared to the non-spatial alternatives. The application to the DTI data set of cocaine users demonstrates the novelty of this model for detecting clinically meaningful regions of differences.

The current work primarily focuses on finding between-region differences in the brain at a single time-point (baseline). Temporally dependent matrix-variate data are also studied \citep{smith2007distribution}, and corresponding spatiotemporal extensions of our model are possible, though non-trivial. Extensions to incorporate covariates (i.e., socio-demographics, such as age, gender, etc.) may be possible via including a regression term in the full conditional distribution of the cluster labels.

\begin{center}
{\large\bf Acknowledgements}
\end{center}
We thank Institute for Drug and Alcohol Studies of Virginia Commonwealth University (VCU) for providing the cocaine users data set \citep{ma2017preliminary} and data manipulation instructions.

\newpage
\begin{center}
	{\large\bf SUPPLEMENTAL MATERIALS}
\end{center}

\section*{Appendix A: Density Functions}
The PDFs of Parameterized Wishart and Inverse Wishart Distribution are given as 
\begin{description}
	\item[The PDF of Wishart Distribution $\bm{X}\sim\mathcal{W}_p(\bm{V},n)$:] $$f(\bm{X}|\bm{V},n)=\displaystyle {\displaystyle {\frac {1}{2^{np/2}\left|{\mathbf {V}/n }\right|^{n/2}\Gamma _{p}\left({\frac {n}{2}}\right)}}{\left|\mathbf {X} \right|}^{(n-p-1)/2}e^{-(1/2)\operatorname {tr} ({[\mathbf {V}/n] }^{-1}\mathbf {X} )}}$$
	\item[The PDF of Inverse Wishart Distribution $\bm{X}\sim\mathcal{IW}_p({\mathbf {\Psi } },\nu)$:]
	$$\displaystyle f({\mathbf {X} }|{\mathbf {\Psi } },\nu )={\frac {\left|{(\nu-p-1)\mathbf {\Psi } }\right|^{\nu /2}}{2^{\nu p/2}\Gamma _{p}({\frac {\nu }{2}})}}\left|\mathbf {X} \right|^{-(\nu +p+1)/2}e^{-{\frac {1}{2}}\operatorname {tr} ({(\nu-p-1)\mathbf {\Psi } }\mathbf {X} ^{-1})}$$
\end{description}

\section*{Appendix B: Variograms}
In this section, we give the details of derivations of variogram. We first give the variogram in the Single-Subject Model below:

\begin{equation}
\begin{aligned}
\mathbb{E}||\bm{A}_u-\bm{A}_v||_F^2&=Tr(\mathbb{E}[(\bm{A}_u-\bm{A}_v)(\bm{A}_u-\bm{A}_v)] )\\
&=Tr(\mathbb{E}\bm{A}_u\bm{A}_u)+Tr(\mathbb{E}\bm{A}_v\bm{A}_v)-2Tr(\mathbb{E}\bm{A}_u\bm{A}_v)\\
&=[Tr(\mathbb{E}\bm{A}_u\bm{A}_u)+Tr(\mathbb{E}\bm{A}_v\bm{A}_v)-2Tr(\mathbb{E}[\bm{A}_u\bm{A}_v|g_u\not=g_v]) ]\times P(g_u\not=g_v|\beta)\\
&\text{(because $\bm{A}_u\bm{A}_u:=[\bm{A}_u\bm{A}_v|g_u=g_v]$)}\\
&=\gamma(m,\nu, \bm{\Sigma})P(g_u\not=g_v|\beta)
\end{aligned}
\end{equation}
Obviously, $\mathbb{E}||\bm{A}_{iu}-\bm{A}_{iv}||_F^2$ can be derived in the same way. Next, we give the the explicit expression of $\gamma(m,\nu, \bm{\Sigma})$. We first have
\begin{equation}
\begin{aligned}
Tr(\mathbb{E}[\bm{A}_u\bm{A}_v|g_u\not=g_v])&=\mathbb{E}[Tr(\mathbb{E}[\bm{A}_u\bm{A}_v|g_u\not=g_v,\bm{M}_{g_u},\bm{M}_{g_v}])]\\
&=\mathbb{E}[Tr(\mathbb{E}[\bm{A}_v|\bm{M}_{g_v}]\mathbb{E}[\bm{A}_u|\bm{M}_{g_u}])]\\
&=\sum_{i=1}^p\sigma_i(\bm{\Sigma})^2=\lambda
\end{aligned}
\end{equation}
where $\sigma_i(.)$ returns the $i$-th eigenvalue of the input function.

The term $Tr(\mathbb{E}\bm{A}_u\bm{A}_u)$ is complex. \citet{gupta1999matrix}[Section 3.3.6] provides trace moments of Wishart and inverse Wishart distribution. 

\begin{theorem}
	Let $\bm{S}\sim\mathcal{W}_p(\bm{\Sigma},m)$ or $\bm{W}\sim\mathcal{IW}_p(\bm{M},m)$. Also, $\bm{S}^{-1}=\bm{W}$ and $\bm{M}=\bm{\Sigma}^{-1}$. $p$ is the matrix dimension. Then we have
	
	1) $\mathbb{E}\bm{W}\bm{W}=(c_1+c_2)\bm{M}\bm{M}(m-p-1)^2+c_2Tr(\bm{M})\bm{M}(m-p-1)^2$
	
	2) $\mathbb{E}\bm{S}\bm{S}=\frac{m+1}{m}\bm{\Sigma}\bm{\Sigma}+\frac{1}{m}Tr(\bm{\Sigma})\bm{\Sigma}$,
	
	3) $\mathbb{E}Tr(\bm{S})\bm{S}=2\frac{1}{m}\bm{\Sigma}\bm{\Sigma}+Tr(\bm{\Sigma})\bm{\Sigma}$
	
	where $c_1=(m-p-2)c_2$ and $c_2=\frac{1}{(m-p)(m-p-1)(m-p-3)}$
\end{theorem}
\begin{proof}
	See \citet{gupta1999matrix}[Section 3.3.6]
\end{proof}

Then we first can obtain 
\begin{equation}
\begin{aligned}
\mathbb{E}[\bm{A}_u\bm{A}_u|\bm{M}_{g_u}=\bm{M}_{g_v}=\bm{M}]&=(c_1+c_2)\bm{M}\bm{M}(m-p-1)^2+c_2Tr(\bm{M})\bm{M}(m-p-1)^2=\lambda^*
\end{aligned}
\end{equation}

Next, we have 
Then we can obtain 
\begin{equation}
\begin{aligned}
Tr(\mathbb{E}\bm{A}_u\bm{A}_u)=Tr(\mathbb{E}\lambda^*)&=\{(c_1+c_2)(\nu+1)\nu^{-1}\sum_{i=1}^p\sigma_i(\bm{\Sigma})^2+(c_1+c_2)\nu^{-1}\sum_{i=1}^p\sigma_i(\bm{\Sigma})\sum_{i=1}^p\sigma_i(\bm{\Sigma})\\
&+c_2(2\nu^{-1}\sum_{i=1}^p\sigma_i(\bm{\Sigma})^2+\sum_{i=1}^p\sigma_i(\bm{\Sigma})\sum_{i=1}^p\sigma_i(\bm{\Sigma}))) \}(m-p-1)^2\\
&=\sigma
\end{aligned}
\end{equation}
In summary, $\gamma(m,\nu, \bm{\Sigma})=2(\sigma-\lambda)$.

\section*{Appendix C: The Joint Probability Mass Density (PMF) of $\{\bm{g}_1, \bm{g}_2, ..., \bm{g}_N\}\cup\{\bm{h}_0, \bm{h}_1\}$}
In this section, we show the expression of the PMF and validate that the PMF is valid:
\begin{equation}
\begin{aligned}
&P(\bm{h},\bm{g})\propto\exp\left[\sum_{i=1}^N\sum_{v=1}^n\alpha\mathcal{I}(g_{iv}=h_{x_iv})+\sum_{x=0}^1\sum_{u\sim v}\beta\mathcal{I}(h_{xu}=h_{xv})+\sum_{i=1}^N\sum_{u\sim v}\left(\beta\mathcal{I}(g_{iu}=g_{iv})-g_{iu}^\xi\right)\right]\\
&=\exp(U(\bm{g},\bm{h},\bm{\theta}))
\end{aligned}
\end{equation}
where $u\sim v$ means $u$ and $v$ are connected. It is obvious that the probability is positive and satisfies the pairwise
Markov property stated in the Hammersley and Clifford Theorem. The normalizing constant $Z(\alpha,\beta,\xi)=\sum_{\bm{g},\bm{h}}\exp(U(\bm{g},\bm{h},\bm{\theta}))$ is intractable. Since the summation is over finite and discrete indices, we have that $0<Z(\alpha,\beta,\xi)<\infty$, revealing that $P(\bm{h},\bm{g})$ is proper.

\section*{Appendix D: The Statistical Role of $h_{xv}$}
\begin{description}
	\item[Step 1: Marginalizing $g_{iv}$]
	\begin{equation}
	\begin{aligned}
	&\left[\bm{A}_{iv}|\{g_u: u\in N_v\},h_{x_iv},\{\bm{V}_{k}:k\},m\right]=\sum_{k=1}^K P(g_{iv}=k|.)\mathcal{IW}_p(\bm{V}_{k},m)\\
	&=\sum_{k=1}^K \underbrace{C_1}_{\substack{Normalizing\\ Constant}}\exp\left[-k^{\xi}+ \beta \sum_{u\in N_v}\mathcal{I}(g_{iu}=k) +\alpha \mathcal{I}(h_{x_iv}=k) \right]\mathcal{IW}_p(\bm{V}_{k},m)\\
	\end{aligned}
	\end{equation}
	\item[Step 2: Marginalizing $\{g_u: u\in N_v\}$]
	\begin{equation}
	\begin{aligned}
	&Q=\sum_{\substack{g_{iu}=1:K,\\u\in N_v}} \exp\left[\beta \sum_{u\in N_v} \mathcal{I}(g_{iu}=k)\right]\underbrace{\mathcal{P}_{g_{iu},u\in N_v}}_{\substack{joint\ p.m.f\\of\ g_{iu},u\in N_v}}\\
	&\left[\bm{A}_{iv}|h_{x_iv},\{\bm{V}_{k}:k\},m\right]=\sum_{k=1}^K \underbrace{C_1 Q \exp\left[-k^{\xi}+\alpha \mathcal{I}(h_{x_iv}=k) \right]}_{\Phi_k(h_{x_iv})}\mathcal{IW}_p(\bm{V}_{k},m)\\
	\end{aligned}
	\end{equation}
\end{description}

\section*{Appendix E: Codes}
The codes and example scripts are available at \url{https://github.com/ZhouLanNCSU/Potts_DTI}.

\bibliographystyle{style}

\begin{thebibliography}{40}
\newcommand{\enquote}[1]{``#1''}
\expandafter\ifx\csname natexlab\endcsname\relax\def\natexlab#1{#1}\fi

\bibitem[{Alexander et~al.(2007)Alexander, Lee, Lazar, and
  Field}]{alexander2007diffusion}
Alexander, A.~L., Lee, J.~E., Lazar, M.,  and Field, A.~S. (2007),
  \enquote{{Diffusion tensor imaging of the brain},}
  \textit{Neurotherapeutics}, 4, 316--329.

\bibitem[{Basser et~al.(1994)Basser, Mattiello, and LeBihan}]{basser1994mr}
Basser, P.~J., Mattiello, J.,  and LeBihan, D. (1994), \enquote{{MR diffusion
  tensor spectroscopy and imaging},} \textit{Biophysical Journal}, 66,
  259--267.

\bibitem[{Cherian et~al.(2016)Cherian, Morellas, and
  Papanikolopoulos}]{cherian2016bayesian}
Cherian, A., Morellas, V.,  and Papanikolopoulos, N. (2016), \enquote{Bayesian
  nonparametric clustering for positive definite matrices,} \textit{IEEE
  Transactions on Pattern Analysis \& Machine Intelligence}, 1--1.

\bibitem[{Clifford(1990)}]{clifford1990markov}
Clifford, P. (1990), \enquote{{Markov} random fields in statistics,}
  \textit{Disorder in physical systems: A volume in honour of John M.
  Hammersley}, 19.

\bibitem[{Cohen et~al.(2017)Cohen, Daw, Engelhardt, Hasson, Li, Niv, Norman,
  Pillow, Ramadge, Turk-Browne, et~al.}]{cohen2017computational}
Cohen, J.~D., Daw, N., Engelhardt, B., Hasson, U., Li, K., Niv, Y., Norman,
  K.~A., Pillow, J., Ramadge, P.~J., Turk-Browne, N.~B. et~al. (2017),
  \enquote{Computational approaches to {fMRI} analysis,} \textit{Nature
  Neuroscience}, 20, 304.

\bibitem[{Cressie(1992)}]{cressie1992statistics}
Cressie, N. (1992), \enquote{Statistics for spatial data,} \textit{Terra Nova},
  4, 613--617.

\bibitem[{Dryden et~al.(2009)Dryden, Koloydenko, and Zhou}]{dryden2009non}
Dryden, I.~L., Koloydenko, A.,  and Zhou, D. (2009), \enquote{Non-{Euclidean}
  statistics for covariance matrices, with applications to diffusion tensor
  imaging,} \textit{The Annals of Applied Statistics}, 1102--1123.

\bibitem[{Ennis and Kindlmann(2006)}]{ennis2006orthogonal}
Ennis, D.~B.,  and Kindlmann, G. (2006), \enquote{Orthogonal tensor invariants
  and the analysis of diffusion tensor magnetic resonance images,}
  \textit{Magnetic Resonance in Medicine: An Official Journal of the
  International Society for Magnetic Resonance in Medicine}, 55, 136--146.

\bibitem[{Goldsmith et~al.(2014)Goldsmith, Huang, and
  Crainiceanu}]{goldsmith2014smooth}
Goldsmith, J., Huang, L.,  and Crainiceanu, C.~M. (2014), \enquote{Smooth
  scalar-on-image regression via spatial {Bayesian} variable selection,}
  \textit{Journal of Computational and Graphical Statistics}, 23, 46--64.

\bibitem[{Gupta and Nagar(1999)}]{gupta1999matrix}
Gupta, A.~K.,  and Nagar, D.~K. (1999), \textit{Matrix variate distributions},
  vol. 104, CRC Press.

\bibitem[{Heidelberger and Welch(1981)}]{heidelberger1981spectral}
Heidelberger, P.,  and Welch, P.~D. (1981), \enquote{A spectral method for
  confidence interval generation and run length control in simulations,}
  \textit{Communications of the ACM}, 24, 233--245.

\bibitem[{Johnson et~al.(2013)Johnson, Liu, Bartsch, and
  Nichols}]{johnson2013bayesian}
Johnson, T.~D., Liu, Z., Bartsch, A.~J.,  and Nichols, T.~E. (2013), \enquote{A
  {Bayesian} non-parametric {Potts} model with application to pre-surgical
  {fMRI} data,} \textit{Statistical Methods in Medical Research}, 22, 364--381.

\bibitem[{Kang et~al.(2011)Kang, Johnson, Nichols, and Wager}]{kang2011meta}
Kang, J., Johnson, T.~D., Nichols, T.~E.,  and Wager, T.~D. (2011),
  \enquote{Meta analysis of functional neuroimaging data via Bayesian spatial
  point processes,} \textit{Journal of the American Statistical Association},
  106, 124--134.

\bibitem[{Lan et~al.(2016)Lan, Zhao, Kang, and Yu}]{lan2016bayesian}
Lan, Z., Zhao, Y., Kang, J.,  and Yu, T. (2016), \enquote{Bayesian network
  feature finder ({BANFF}): an {R} package for gene network feature selection,}
  \textit{Bioinformatics}, 32, 3685--3687.

\bibitem[{Lane et~al.(2010)Lane, Steinberg, Ma, Hasan, Kramer, Zuniga,
  Narayana, and Moeller}]{lane2010diffusion}
Lane, S.~D., Steinberg, J.~L., Ma, L., Hasan, K.~M., Kramer, L.~A., Zuniga,
  E.~A., Narayana, P.~A.,  and Moeller, F.~G. (2010), \enquote{Diffusion tensor
  imaging and decision making in cocaine dependence,} \textit{PLoS One}, 5,
  e11591.

\bibitem[{Lee and Schwartzman(2017)}]{lee2017inference}
Lee, H.~N.,  and Schwartzman, A. (2017), \enquote{Inference for eigenvalues and
  eigenvectors in exponential families of random symmetric matrices,}
  \textit{Journal of Multivariate Analysis}, 162, 152--171.

\bibitem[{Li et~al.(2018)Li, Wang, Liang, Yi, Xie, Gazdar, and
  Xiao}]{li2018bayesian}
Li, Q., Wang, X., Liang, F., Yi, F., Xie, Y., Gazdar, A.,  and Xiao, G. (2018),
  \enquote{A {Bayesian} hidden {Potts} mixture model for analyzing lung cancer
  pathology images,} \textit{Biostatistics}.

\bibitem[{Liang(2010)}]{liang2010double}
Liang, F. (2010), \enquote{A double {Metropolis-Hastings} sampler for spatial
  models with intractable normalizing constants,} \textit{Journal of
  Statistical Computation and Simulation}, 80, 1007--1022.

\bibitem[{Lindsay and Lesperance(1995)}]{lindsay1995review}
Lindsay, B.~G.,  and Lesperance, M.~L. (1995), \enquote{A review of
  semiparametric mixture models,} \textit{Journal of Statistical Planning and
  Inference}, 47, 29--39.

\bibitem[{Liu et~al.(2014)Liu, Awate, Anderson, and
  Fletcher}]{liu2014functional}
Liu, W., Awate, S.~P., Anderson, J.~S.,  and Fletcher, P.~T. (2014), \enquote{A
  functional network estimation method of resting-state {fMRI} using a
  hierarchical {Markov} random field,} \textit{NeuroImage}, 100, 520--534.

\bibitem[{Lo et~al.(2010)Lo, Wang, Chou, Wang, He, and Lin}]{lo2010diffusion}
Lo, C.-Y., Wang, P.-N., Chou, K.-H., Wang, J., He, Y.,  and Lin, C.-P. (2010),
  \enquote{Diffusion tensor tractography reveals abnormal topological
  organization in structural cortical networks in {Alzheimer's} disease,}
  \textit{Journal of Neuroscience}, 30, 16876--16885.

\bibitem[{Ma et~al.(2009)Ma, Hasan, Steinberg, Narayana, Lane, Zuniga, Kramer,
  and Moeller}]{ma2009diffusion}
Ma, L., Hasan, K.~M., Steinberg, J.~L., Narayana, P.~A., Lane, S.~D., Zuniga,
  E.~A., Kramer, L.~A.,  and Moeller, F.~G. (2009), \enquote{Diffusion tensor
  imaging in cocaine dependence: regional effects of cocaine on corpus callosum
  and effect of cocaine administration route,} \textit{Drug and Alcohol
  Dependence}, 104, 262--267.

\bibitem[{Ma et~al.(2017)Ma, Steinberg, Wang, Schmitz, Boone, Narayana, and
  Moeller}]{ma2017preliminary}
Ma, L., Steinberg, J.~L., Wang, Q., Schmitz, J.~M., Boone, E.~L., Narayana,
  P.~A.--- (2017), \enquote{A preliminary longitudinal study of white matter
  alteration in cocaine use disorder subjects,} \textit{Drug and Alcohol
  Dependence}, 173, 39--46.

\bibitem[{McCullagh et~al.(2008)McCullagh, Yang, et~al.}]{mccullagh2008many}
McCullagh, P., Yang, J. et~al. (2008), \enquote{How many clusters?}
  \textit{Bayesian Analysis}, 3, 101--120.

\bibitem[{Muschelli et~al.(2014)Muschelli, Sweeney, and
  Crainiceanu}]{muschelli2014brainr}
Muschelli, J., Sweeney, E.,  and Crainiceanu, C. (2014), \enquote{brainR:
  Interactive 3 and 4D Images of High Resolution Neuroimage Data,} \textit{The
  R Journal}, 6, 41.

\bibitem[{Musgrove et~al.(2016)Musgrove, Hughes, and Eberly}]{musgrove2016fast}
Musgrove, D.~R., Hughes, J.,  and Eberly, L.~E. (2016), \enquote{Fast, fully
  Bayesian spatiotemporal inference for fMRI data,} \textit{Biostatistics}, 17,
  291--303.

\bibitem[{Park and Haran(2018)}]{park2018bayesian}
Park, J.,  and Haran, M. (2018), \enquote{Bayesian inference in the presence of
  intractable normalizing functions,} \textit{Journal of the American
  Statistical Association}, 113, 1372--1390.

\bibitem[{Rand(1971)}]{rand1971objective}
Rand, W.~M. (1971), \enquote{Objective criteria for the evaluation of
  clustering methods,} \textit{Journal of the American Statistical
  Association}, 66, 846--850.

\bibitem[{Reich et~al.(2018)Reich, Guinness, Vandekar, Shinohara, and
  Staicu}]{reich2018fully}
Reich, B.~J., Guinness, J., Vandekar, S.~N., Shinohara, R.~T.,  and Staicu,
  A.-M. (2018), \enquote{Fully Bayesian spectral methods for imaging data,}
  \textit{Biometrics}, 74, 645--652.

\bibitem[{Reich and Shaby(2018)}]{reich2018spatial}
Reich, B.~J.,  and Shaby, B.~A. (2018), \enquote{A spatial {Markov} model for
  climate extremes,} \textit{Journal of Computational and Graphical
  Statistics}.

\bibitem[{Robert(2007)}]{robert2007bayesian}
Robert, C. (2007), \textit{The {Bayesian} choice: {From} decision-theoretic
  foundations to computational implementation}, Springer Science \& Business
  Media.

\bibitem[{Schwartzman et~al.(2008)Schwartzman, Mascarenhas, and
  Taylor}]{schwartzman2008inference}
Schwartzman, A., Mascarenhas, W.~F.,  and Taylor, J.~E. (2008),
  \enquote{Inference for eigenvalues and eigenvectors of Gaussian symmetric
  matrices,} \textit{The Annals of Statistics}, 2886--2919.

\bibitem[{Smith and Garth(2007)}]{smith2007distribution}
Smith, P.~J.,  and Garth, L.~M. (2007), \enquote{Distribution and
  characteristic functions for correlated complex {Wishart} matrices,}
  \textit{Journal of Multivariate Analysis}, 98, 661--677.

\bibitem[{Spence et~al.(2007)Spence, Carmack, Gunst, Schucany, Woodward, and
  Haley}]{spence2007accounting}
Spence, J.~S., Carmack, P.~S., Gunst, R.~F., Schucany, W.~R., Woodward, W.~A.,
  and Haley, R.~W. (2007), \enquote{Accounting for spatial dependence in the
  analysis of {SPECT} brain imaging data,} \textit{Journal of the American
  Statistical Association}, 102, 464--473.

\bibitem[{Woolrich et~al.(2004{\natexlab{a}})Woolrich, Behrens, Beckmann,
  Jenkinson, and Smith}]{woolrich2004multilevel}
Woolrich, M.~W., Behrens, T.~E., Beckmann, C.~F., Jenkinson, M.,  and Smith,
  S.~M. (2004{\natexlab{a}}), \enquote{Multilevel linear modelling for {fMRI}
  group analysis using {Bayesian} inference,} \textit{NeuroImage}, 21,
  1732--1747.

\bibitem[{Woolrich et~al.(2004{\natexlab{b}})Woolrich, Jenkinson, Brady, and
  Smith}]{woolrich2004fully}
Woolrich, M.~W., Jenkinson, M., Brady, J.~M.--- (2004{\natexlab{b}}),
  \enquote{Fully Bayesian spatio-temporal modeling of fMRI data,} \textit{IEEE
  transactions on medical imaging}, 23, 213--231.

\bibitem[{Wu(1982)}]{wu1982potts}
Wu, F.-Y. (1982), \enquote{The {Potts} model,} \textit{Reviews of Modern
  Physics}, 54, 235.

\bibitem[{Wu et~al.(2013)Wu, Stramaglia, Chen, Liao, and
  Marinazzo}]{wu2013mapping}
Wu, G.-R., Stramaglia, S., Chen, H., Liao, W.,  and Marinazzo, D. (2013),
  \enquote{Mapping the voxel-wise effective connectome in resting state
  {fMRI},} \textit{PloS one}, 8, e73670.

\bibitem[{Xue et~al.(2018)Xue, Bowman, and Kang}]{xue2018bayesian}
Xue, W., Bowman, F.~D.,  and Kang, J. (2018), \enquote{A {Bayesian} Spatial
  Model to Predict Disease Status Using Imaging Data From Various Modalities,}
  \textit{Frontiers in neuroscience}, 12, 184.

\bibitem[{Zhao et~al.(2014)Zhao, Kang, and Yu}]{zhao2014bayesian}
Zhao, Y., Kang, J.,  and Yu, T. (2014), \enquote{A {Bayesian} nonparametric
  mixture model for selecting genes and gene subnetworks,} \textit{The Annals
  of Applied Statistics}, 8, 999.

\end{thebibliography}


\end{document}